\documentclass[12pt, reqno]{amsart}
\usepackage{amssymb}
\usepackage{amsmath, mathrsfs}
\allowdisplaybreaks[1]
\usepackage{graphicx,color}
\usepackage{wrapfig,framed}

\usepackage{extpfeil}

\usepackage{color}

\usepackage[height=22.7cm, width=17cm, hmarginratio={1:1}]{geometry}
\usepackage{hyperref}

\usepackage{mathtools}

\usepackage{mathdots}

\usepackage{lscape}

\usepackage{makecell}

\usepackage{tikz}
\usetikzlibrary{arrows.meta}
\usetikzlibrary{decorations.markings}

\usepackage{caption}

\theoremstyle{plain}
\newtheorem{theorem}{Theorem}
\newtheorem{prop}{Proposition}
\newtheorem{lemma}{Lemma}

\newtheorem*{theorem'}{Theorem\ 1'}
\theoremstyle{definition}

\newtheorem*{remark}{Remark}

\newcommand{\beq}{\begin{equation}}
\newcommand{\eeq}{\end{equation}}
\newcommand{\nn}{\nonumber}

\newcommand{\F}{\mathcal{F}}

\newcommand{\QQ}{{\mathbb Q}}

\newcommand{\CC}{{\mathbb C}}
\newcommand{\ZZ}{{\mathbb Z}}

\newcommand{\bT}{{\bf T}}
\newcommand{\e}{\epsilon}
\newcommand{\p}{\partial}

\newcommand{\diag}{{\rm diag}}
\newcommand{\TT}{{\mathcal{T}}}
\newcommand{\Tr}{{\rm Tr}}

\title{On enumeration of $l$-hypermaps}
\author{Zhengfei Huang, Di Yang}
\address{School of Mathematical Sciences, University of Science and Technology of China, Hefei 230026, P.R. China}
\email{huangzf063@mail.ustc.edu.cn, diyang@ustc.edu.cn}
\date{}

\begin{document}

\begin{abstract}
The $l$-hypermaps, $l\ge2$, which generalize (dual of) ribbon graphs ($l=2$ case), 
are interesting enumerative objects. 
In this paper, based on a theorem of Carlet--van de Leur--Posthuma--Shadrin and the matrix-resolvent method, we derive an explicit formula for $k$-point generating series of enumeration of 
$l$-hypermaps, which generalizes the one obtained in~\cite{DY1} for the $l=2$ case.
We also generalize a theorem of Dubrovin~\cite{Du2}.
\end{abstract}

\maketitle


\section{Introduction}\label{section1}
Ribbon graphs are graphs drawn on a Riemann surface. 
Their enumeration is a very important question in combinatorics, random matrix theory, 
and mathematical physics (cf.~e.g.~\cite{AMM, BIZ,EMP,HZ,DGZ,GMMMO,Hooft, Hooft2, KKN,Mehta, Tutte,Witten}). 
Closed formula with 1 vertex was obtained 
by Harer and Zagier~\cite{HZ} (cf.~also e.g.~\cite{DY1,WL}). 
Explicit formulas with 2 vertices were obtained in \cite{DY1, MS, Zhou}.
Explicit formulas with more than 2 vertices were obtained in \cite{DY1, Zhou}. 
We also refer to e.g.~\cite{AvM, Deift, DY1, GMMMO, STY, Y20, Zhou} for relations to integrable systems,  
to e.g.~\cite{ACKM, BD, DY1, DY2, Ercolani, EP, Zhougenus} for explicit formulas with a fixed genus, 
and 
to e.g.~\cite{CEO, Deift, DM, EMP, Eynard, EO, FIK, KKN, MulaseS} for analytic considerations including Painlev\'e transcendents and topological recursion.

By vertex-face duality for a ribbon graph, enumeration of ribbon graphs is equivalent to enumeration of polygon-angulations of a Riemann surface.
The latter is a special case of counting the so-called $l$-hypermaps.
Following Do and Manescu~\cite{DM}, an {\it $l$-hypermap of type $(g,k,d)$} is 
a triple $(\sigma_0,\sigma_1,\sigma_2)$, with $\sigma_0,\sigma_1,\sigma_2$ being elements in the symmetric group $S_d$, satisfying
\begin{itemize}
\item[(i)] $\sigma_0\sigma_1\sigma_2=$id;
\item[(ii)] $\sigma_2$ consists of $k$ disjoint cycles;
\item[(iii)] all cycles of~$\sigma_1$ have length $l$; 
\item[(iv)] $\sigma_0$ has $v$ cycles, where $v=2-2g-k+(l-1)\frac{d}{l}$.
\end{itemize}
An $l$-hypermap is called {\it connected} if the subgroup generated by $\sigma_0,\sigma_1,\sigma_2$ acts transitively on $\{1,\dots,d\}$. An $l$-hypermap is called {\it labelled} if the disjoint cycles of $\sigma_2$ are labelled from $1$ to $k$.
A connected $l$-hypermap $(\sigma_0,\sigma_1,\sigma_2)$ of type $(g,k,d)$ 
can be represented graphically~\cite{DM} drawn on a Riemann surface of genus $g$ 
 as follows:  each cycle of $\sigma_1$ 
corresponds to a blue $l$-gon (called a {\it hyperedge}),
each cycle of $\sigma_2$ corresponds to a white polygon of size being the length of the cycle,
and each cycle of $\sigma_0$ corresponds to a vertex whose valency of hyperedges equals the length of the cycle.
When $l=2$ all hyperedges are ribbons, so $2$-hypermaps are ribbon graphs.
The following are two examples of graphical representations of $3$-hypermaps:
Figure~1 shows a $3$-hypermap of type $(0,2,12)$, 
which has $8$ vertices, $2$ white faces and $4$ hyperedges, and
Figure~2 shows a $3$-hypermap of type $(1,2,12)$, 
which has $6$ vertices, $2$ white faces and $4$ hyperedges.
If one allows $\sigma_1$ to have an arbitrary cycle type, then one gets more general 
hypermaps~\cite{Cori,WL}. 

\begin{minipage}[t]{0.52\textwidth}
\centering
\begin{tikzpicture}[>=Stealth]
\coordinate (p1) at (0,0);
\coordinate (p2) at (4,0);
\coordinate (p3) at (2,1.6);
\coordinate (p4) at (2,3.5);
\coordinate (p5) at (-0.5,4);
\coordinate (p6) at (4.5,4);
\coordinate (p7) at (6,4);
\coordinate (p8) at (5.25,2.5);
\filldraw[fill=cyan!50](p1)--(p2)--(p3)--cycle;
\filldraw[fill=cyan!50](p1)--(p4)--(p5)--cycle;
\filldraw[fill=cyan!50](p2)--(p4)--(p6)--cycle;
\filldraw[fill=cyan!50](p6)--(p7)--(p8)--cycle;
\fill (p1) circle (2pt);
\fill (p2) circle (2pt);
\fill (p3) circle (2pt);
\fill (p4) circle (2pt);
\fill (p5) circle (2pt);
\fill (p6) circle (2pt);
\fill (p7) circle (2pt);
\fill (p8) circle (2pt);
\draw[color=black](2,2.2)circle(0.2cm);
\node at (2,2.2){\tiny\color{black}1};
\draw[color=black](5.25,1.25)circle(0.2cm);
\node at (5.25,1.25){\tiny\color{black}2};
\end{tikzpicture}\label{figure1}

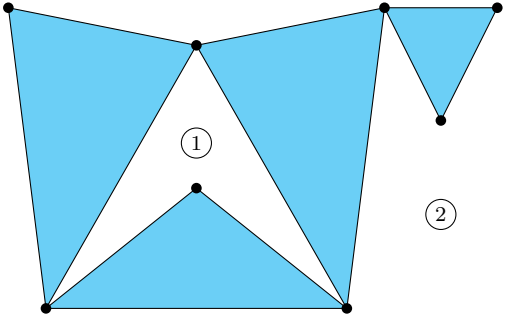
\captionof{figure}{a graphical representation of a $3$-hypermap on a sphere.}
\end{minipage}
\hfill
\begin{minipage}[t]{0.52\textwidth}
\centering
\begin{tikzpicture}[>=Stealth]
\coordinate (p1) at (0.3,0.8);
\coordinate (p2) at (2.7,0.8);
\coordinate (p3) at (0.3,3.2);
\coordinate (p4) at (2.7,3.2);
\coordinate (p5) at (1.5,0);
\coordinate (p6) at (-0.5,2);
\coordinate (p7) at (3.5,2);
\coordinate (p8) at (1.5,4);
\filldraw[fill=cyan!50](p1)--(p5)--(p2)--cycle;
\filldraw[fill=cyan!50](p2)--(p4)--(p7)--cycle;
\filldraw[fill=cyan!50](p3)--(p4)--(p8)--cycle;
\filldraw[fill=cyan!50](p1)--(p3)--(p6)--cycle;
\fill (p1) circle (2pt);
\fill (p2) circle (2pt);
\fill (p3) circle (2pt);
\fill (p4) circle (2pt);
\fill (p5) circle (2pt);
\fill (p6) circle (2pt);
\fill (p7) circle (2pt);
\fill (p8) circle (2pt);
\draw[color=black](1.5,2)circle(0.2cm);
\node at (1.5,2){\tiny\color{black}1};
\draw[color=black](-0.1,0.4)circle(0.2cm);
\node at (-0.1,0.4){\tiny\color{black}2};
\draw[-latex](-0.5,0)--(0.6,0);
\draw (0.6,0)--(p5);
\draw[-latex](p5)--(2.6,0);
\draw (2.6,0)--(3.5,0);
\draw[-latex](3.5,0)--(3.5,1.1);
\draw (3.5,1.1)--(p7);
\draw[-latex](p7)--(3.5,3.1);
\draw (3.5,3.1)--(3.5,4);
\draw[-latex](p8)--(2.6,4);
\draw (2.6,4)--(3.5,4);
\draw[-latex](-0.5,4)--(0.6,4);
\draw (0.6,4)--(p8);
\draw[-latex](-0.5,0)--(-0.5,1.1);
\draw (-0.5,1.1)--(p6);
\draw[-latex](p6)--(-0.5,3.1);
\draw (-0.5,3.1)--(-0.5,4);
\end{tikzpicture}\label{figure2}

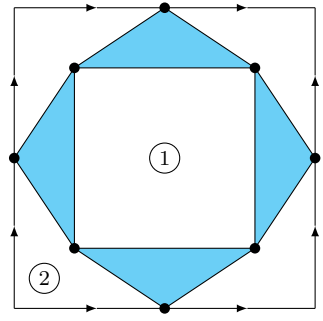
\captionof{figure}{a graphical representation of a $3$-hypermap on a torus.}
\end{minipage}

A {\it conjugation equivalence} from a labelled $l$-hypermap
 $(\sigma_0,\sigma_1,\sigma_2)$ to another one $(\tau_0,\tau_1,\tau_2)$ is a conjugation that simultaneously transforms 
 $\sigma_i$ to $\tau_i$, $i=0,1,2$, and preserves the labels. An automorphism of a labelled $l$-hypermap is a {\it conjugation equivalence} from the  $l$-hypermap to itself. 
Let $M_{g,k}^{[l]}(b_1,\dots,b_k)$ be the weighted count of connected, labelled $l$-hypermaps $(\sigma_0,\sigma_1,\sigma_2)$ of type $(g,k,|b|)$, where the cycle of $\sigma_2$ labelled $i$ has length $b_i$, for $i=1,\dots,k$, and $|b|=b_1+\dots+b_k$. 
The weight for  
an $l$-hypermap is equal to the reciprocal of the number of its automorphisms.
For example, when $l=5$, $g=0$, $k=3$ and $(b_1,b_2,b_3)=(1,2,2)$, there are four such 
$5$-hypermaps (see Figure 3)
and all their automorphisms are trivial.
So $M_{0,3}^{[5]}(1,2,2)=4$. 
(We shall notice that for the case when $l=2$, there is a systematic proportional factor 
$\prod_{j=1}^k b_j$ between the number in~\cite{DY1} and $M_{g,k}^{[l]}(b_1,\dots,b_k)$ here; 
see e.g. the footnote of~\cite{DY2}.)

{\center{
\begin{tikzpicture}[>=Stealth]
\coordinate (p1) at (0,0);
\coordinate (p2) at (-1.5,1.5);
\coordinate (p3) at (1.5,1.5);
\fill[cyan!50] (-2,-2) rectangle (2,2);
\filldraw[fill=white](p1)arc(0:90:1.5)arc(180:270:1.5);
\filldraw[fill=white](p1)arc(-90:0:1.5)arc(90:180:1.5);
\filldraw[fill=white](p1)arc(90:450:0.7);
\fill (p1) circle (2pt);
\fill (p2) circle (2pt);
\fill (p3) circle (2pt);
\draw[color=black](0,-0.7)circle(0.2cm);
\node at (0,-0.7){\tiny\color{black}1};
\draw[color=black](-0.75,0.75)circle(0.2cm);
\node at (-0.75,0.75){\tiny\color{black}2};
\draw[color=black](0.75,0.75)circle(0.2cm);
\node at (0.75,0.75){\tiny\color{black}3};
\end{tikzpicture}
\begin{tikzpicture}[>=Stealth]
\coordinate (p1) at (0,0);
\coordinate (p2) at (-1.5,1.5);
\coordinate (p3) at (1.5,1.5);
\fill[cyan!50] (-2,-2) rectangle (2,2);
\filldraw[fill=white](p1)arc(0:90:1.5)arc(180:270:1.5);
\filldraw[fill=white](p1)arc(-90:0:1.5)arc(90:180:1.5);
\filldraw[fill=white](p1)arc(90:450:0.7);
\fill (p1) circle (2pt);
\fill (p2) circle (2pt);
\fill (p3) circle (2pt);
\draw[color=black](0,-0.7)circle(0.2cm);
\node at (0,-0.7){\tiny\color{black}1};
\draw[color=black](-0.75,0.75)circle(0.2cm);
\node at (0.75,0.75){\tiny\color{black}2};
\draw[color=black](0.75,0.75)circle(0.2cm);
\node at (-0.75,0.75){\tiny\color{black}3};
\end{tikzpicture}
\begin{tikzpicture}[>=Stealth]
\coordinate (p1) at (-0.7,0.7);
\coordinate (p2) at (0.5,-0.5);
\coordinate (p3) at (1.7,-1.7);
\fill[cyan!50] (-2,-2) rectangle (2,2);
\filldraw[fill=white](p1)arc(90:0:1.2)arc(90:0:1.2)arc(270:180:1.2)arc(270:180:1.2);
\filldraw[fill=white](p1)arc(-45:315:0.6);
\fill (p1) circle (2pt);
\fill (p2) circle (2pt);
\fill (p3) circle (2pt);
\draw[color=black](-1.124,1.124)circle(0.2cm);
\node at (-1.124,1.124){\tiny\color{black}1};
\draw[color=black](-0.1,0.1)circle(0.2cm);
\node at (-0.1,0.1){\tiny\color{black}3};
\draw[color=black](1.1,-1.1)circle(0.2cm);
\node at (1.1,-1.1){\tiny\color{black}2};
\end{tikzpicture}
\begin{tikzpicture}[>=Stealth]
\coordinate (p1) at (-0.7,0.7);
\coordinate (p2) at (0.5,-0.5);
\coordinate (p3) at (1.7,-1.7);
\fill[cyan!50] (-2,-2) rectangle (2,2);
\filldraw[fill=white](p1)arc(90:0:1.2)arc(90:0:1.2)arc(270:180:1.2)arc(270:180:1.2);
\filldraw[fill=white](p1)arc(-45:315:0.6);
\fill (p1) circle (2pt);
\fill (p2) circle (2pt);
\fill (p3) circle (2pt);
\draw[color=black](-1.124,1.124)circle(0.2cm);
\node at (-1.124,1.124){\tiny\color{black}1};
\draw[color=black](-0.1,0.1)circle(0.2cm);
\node at (-0.1,0.1){\tiny\color{black}2};
\draw[color=black](1.1,-1.1)circle(0.2cm);
\node at (1.1,-1.1){\tiny\color{black}3};
\end{tikzpicture}

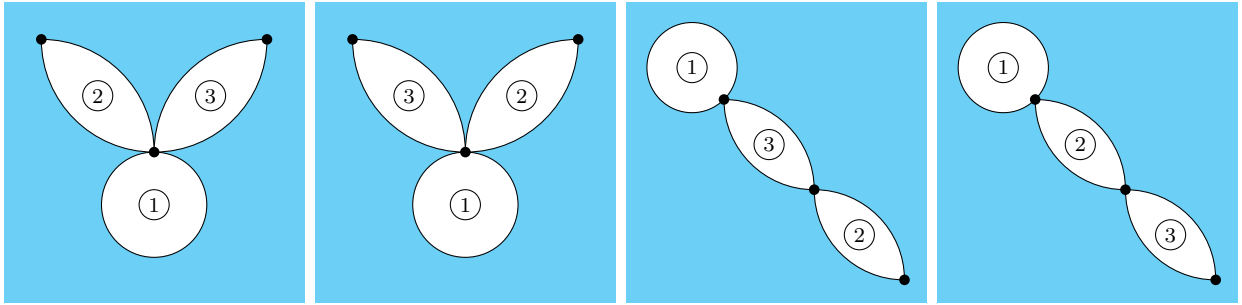
\captionof{figure}{5-hypermaps of type $(0,3,5)$ with $(b_1,b_2,b_3)=(1,2,2)$. }
}}

For $k\ge1$, define the $k$-point functions for the numbers $M^{[l]}_{g,k}(b_1,\dots,b_k)$ by
\beq
\mathcal{C}_{k}(\lambda_1,\dots,\lambda_k;n;l):=
\sum_{b_1,\dots,b_k\ge1}\prod_{j=1}^k\frac{b_j}{\lambda_j^{b_j+1}}M^{[l]}_{k}(b_1,\dots,b_k;n),
\eeq
where
\beq
M^{[l]}_k(b_1,\dots,b_k; n):=\sum_{g=0}^{[(l-1)\frac{|b|}{2l}-\frac{k}{2}+\frac12]}M^{[l]}_{g,k}(b_1,\dots,b_k)n^{2-2g-k+(l-1)\frac{|b|}{l}}.
\eeq

The following theorem, which generalizes~\cite[Theorem~1.1.1]{DY1} for ribbon graphs (the $l=2$ case), 
 will be proved in Section~\ref{conclusion} using the matrix-resolvent method~\cite{BDY1,DY1,FYZ,Y20} and the Carlet--van de Leur--Posthuma--Shadrin theorem~\cite{CLPS} (see~\eqref{CLPSthm} below).
\begin{theorem}\label{cnj1}
For any $l\ge2$, the $k$-point function has the following expression:
\beq\label{1point}
\mathcal{C}_1(\lambda;n;l)=\sum_{m\ge1}\frac{(-1)^{m}(lm)!}{l^{m}m!\lambda^{lm+1}}\sum_{s=0}^{m}(-1)^{s}\binom{m}{s}\binom{n+ls}{lm+1},
\eeq
\begin{align}\label{npoint}
&\mathcal{C}_{k}(\lambda_1,\dots,\lambda_k; n; l)
=-\sum_{\sigma\in S_k/C_k}
\frac{\Tr \,  M \bigl(\lambda_{\sigma(1)},n;l\bigr)\cdots 
 M\bigl(\lambda_{\sigma(k)},n;l\bigr)}{\prod_{i=1}^k (\lambda_{\sigma(i)}-\lambda_{\sigma(i+1)})}
-\delta_{k,2}
\frac{l-1}
{ (\lambda_1-\lambda_2)^2},\quad k\ge2,
\end{align}
where
$M(\lambda,n;l)$ is an $l\times l$ matrix with the $(i,j)$-entry being
\beq\label{Mm1}
 M(\lambda,n;l)_{i,j}=\delta_{i,j}+\left\{
\begin{array}{ll}
y(\lambda,n,i,j;l),& j=1,\dots,l-1,\\
-ny(\lambda,n,i,0;l),& j=l,\\
\end{array}
\right.
\eeq
with
\begin{align}\label{gm1}
y(\lambda,n,i,j;l)=\lambda^{i-j-l}\sum_{m\ge0}\frac{1}{l^{m}m!\lambda^{lm}}
\sum_{s=0}^{m}(-1)^{s}\binom{m}{s}
(n+1-j-ls)_{lm+l-1-i+j}.
\end{align}
Here $(a)_\ell=\Gamma(a+\ell)/\Gamma(a)$ denotes the increasing Pochhammer symbol. 
\end{theorem}

\begin{remark}
We note that
results equivalent to formula~\eqref{1point} 
had been obtained by Zagier~\cite{Zagier}. 
Indeed, the following formula 
\beq\label{Zagierformula}
 \sum_{n\ge1,\,g\ge0}  lm M_{g,1}^{[l]} (lm) n^{1-2g+(l-1)m} Y^{n-1} 
 = \frac{(lm)!}{ m!\, l^{m}}\frac{(1-Y^l)^m}{(1-Y)^{lm+2}}, \quad m\geq1,
\eeq
is a special case of \cite[Theorem~1 and Application~1]{Zagier}, which leads to~\eqref{1point}. 
Moreover, it is shown in~\cite[Application~4]{Zagier} the 
equivalence between~\eqref{Zagierformula} and the following formula
\beq
lmM^{[l]}_{g,1}(lm)
=\frac{(lm)!}{l^m m!(1-2g+(l-1)m)!}[t^{2g}]\Bigl(t^{(l-1)m+2}\frac{(1-e^{-lt})^m}{(1-e^{-t})^{lm+2}}e^{-t}\Bigr),\quad g\ge0,\label{Zagierformula2}
\eeq
as a special case of~\cite[Application~4]{Zagier}. 
We also note that there is a straightforward equivalence between~\eqref{1point} and a special case of a work of Stanley~\cite{Stanley}.
\end{remark}

It turns out that the formal functions appearing in
the above matrix $M(\lambda,n;l)$ can also be expressed in terms of 
 certain algebraic curves introduced in~\cite{DYZ3} (see Theorem~\ref{curve}).
Indeed, following~\cite{DYZ3}, define an odd Laurent series
\beq
 X =\frac{1}{u} - \frac{r-1}6 u + \frac{(r-1)(2r+1)(r-3)}{360} u^3 + \cdots
  \quad\in\;\frac1u\QQ[r][[u^2]] 
\eeq
as the unique solution to 
\beq\label{deftc1}
 \frac{(X+1)^{r+1} - (X-1)^{r+1}}{2\, (r+1)} = \frac1{u^r} 
\eeq
and define coefficients $\widetilde C_s(r,i,j)\in\QQ[r,i,j]$ by
\beq\label{deftc2}
 - (X+1)^i (X-1)^j \frac{dX}{du} = \sum_{s\geq0} \widetilde C_s(r,i,j) u^{s-i-j-2}\;. 
\eeq
For example, 
\beq\label{tCini}
\widetilde C_0(r,i,j) = 1 , \quad \widetilde C_1(r,i,j) =  i-j.
\eeq
The following theorem will be proved in Section~\ref{conclusion}.
\begin{theorem}\label{curve}
The series $y(\lambda,n,i,j;l)$ in Theorem~\ref{cnj1} has the following alternative expression:
\beq\label{ynew}
y(\lambda,n,i,j;l)=
-\delta_{i,j}+\frac{\delta_{i,l}\delta_{j,0}}{n}+
\sum_{m\ge\lceil\frac{i-j}{l-1}\rceil}\frac{2^m(m)_{(l-1)m+j-i}}{(2\lambda)^{lm+j-i}}
\widetilde C_{(l-1)m+j-i}(l-1,n+l-1-i,-n+j-1).
\eeq
Moreover, the $1$-point function can be expressed as
\beq\label{1point2}
\mathcal{C}_1(\lambda;n;l)=\sum_{m\ge1}\frac{(m+1)_{(l-1)m}}{2^{(l-1)m}\lambda^{lm+1}}
\widetilde C_{(l-1)m+1}(l-1,n,-n).
\eeq
\end{theorem}

Define 
the normalized partition function of enumeration of $l$-hypermaps by
\beq
Z^{\rm hyperone}(x,{\bf s};\e;l):=\exp\biggl(\,\sum_{g\ge0,\,k\ge1}\frac{\e^{2g-2}}{k!} \sum_{b_1,\dots,b_k\ge1} x^{2-2g-k+\frac{l-1}{l}|b|}M^{[l]}_{g,k}(b_1,\dots,b_k)\prod_{i=1}^kb_{i} s_{b_i}\biggr),
\eeq
where ${\bf s}=(s_1,s_2,s_3,\dots)$. Its logarithm $\log Z^{\rm hyperone}(x,{\bf s};\e;l)=:\F^{\rm hyperone}(x,{\bf s};\e;l)$
is called the {\it normalized free energy}.
Let us introduce the {\it corrected 
free energy}~$\F=\F(x,{\bf s};\e;l)$ by 
$$
\F(x,{\bf s};\e;l) = \F^{\rm hyperone}(x,{\bf s};\e;l) + \frac{x^2}{2\e^2}\Bigl(\log x-\frac32\Bigr)
-\frac{\log x}{12}
+\sum_{g\ge2}\frac{\e^{2g-2}B_{2g}}{4g(g-1)x^{2g-2}} .
$$
Here $B_j$ denotes the $j$th Bernoulli number.
The exponential $\exp \F =:Z$ is called the {\it corrected partition function}. 
And define the {\it genus $g$ free energy of $l$-hypermap} $\F_g(x,\mathbf{s};l)$ by $\F(x,\mathbf{s};\e;l)=:\sum_{g\ge0}\e^{2g-2}\F_g(x,\mathbf{s};l)$.
For the case when $l=2$, the corrected free energy or partition function
was introduced in~\cite{Du2} (cf.~\cite{AvM, DY1, Yang}), and this partition function 
is closely related~\cite{Du2} (cf.~\cite{CLPS0, CLPS, Yang, Yang25}) 
to the topological partition function of a particular two-dimensional Frobenius manifold with Frobenius potential 
$$
 F=\frac 12(v^1)^2v^0+\frac12(v^0)^2\log v^0-\frac34(v^0)^2.
$$
We refer the reader to~\cite{Du, DZ-norm} about the notion of Frobenius manifold.

For any $l\ge2$, Carlet--van de Leur--Posthuma--Shadrin~\cite{CLPS} proved that the 
partition function $Z^{\rm hyperone}(x,{\bf t};\e;l)|_{x=1}$ is equal to  
the topological partition function (cf.~\cite{DZ-norm, Givental}), aka the total descendant potential, of the $l$-dimensional 
Hurwitz--Frobenius manifold $M_{0;1,l-1}$ (cf.~\cite{AK, Du, LZZ}) restricted to the resonance flows. 
According to~\cite{CLPS, LZZ} (cf.~\cite{CDZ}), 
 after performing a generalized Madelung--Hasimoto transformation and a Miura-type transformation, 
the Dubrovin--Zhang hierarchy of $M_{0;1,l-1}$
becomes the extended constrained KP hierarchy. 
Recall that the constrained KP hierarchy was introduced in \cite{Cheng} (see also~\cite{CL, KSS}), 
and that for the case when $l=2$ the constrained KP hierarchy is the celebrated 
AKNS hierarchy (aka the nonlinear Schr\"odinger hierarchy).
It is known (cf.~\cite{Du, LZZ}) that $M_{0;1,l-1}$ is a Legendre-type transformation of the Frobenius manifold $M(A_{l-1}, 1)$. 
The latter was constructed by 
Dubrovin--Zhang~\cite{DZ-Weyl}, whose Frobenius potential is equal to the genus~0 
primary free energy of the orbifold projective line~$\mathbb{P}^1_{l-1,1}$~\cite{Milanov--Tseng}. 
Let $\mathcal{D}=\mathcal{D}(\bT;\e;l)$, 
$\bT=(T^{a,j})_{a=0,\dots,l-1,j\ge0}$, be the 
topological partition function of $M_{0;1,l-1}$, which satisfies the following string equation
\beq\label{string}
\sum_{a=0}^{l-1}\sum_{j\ge1}T^{a,j}\frac{\partial \log \mathcal{D}}{\partial T^{a,j-1}}
+\frac{1}{2\e^2}T^{0,0}T^{l-1,0}+\frac{1}{2(l-1)\e^2}\sum_{a=1}^{l-2}T^{a,0}T^{l-1-a,0}=\frac{\partial \log \mathcal{D}}{\partial T^{1,0}}.
\eeq

Due to the relation to $M(A_{l-1}, 1)$ and by \cite[Theorem~6]{Y24}, we know that 
\beq
\mathcal{D} = Z^{\{1\}}
\eeq
where $Z^{\{\kappa\}}$, $\kappa=0,\dots,l-1$, denotes the $\kappa$th partition function of $M(A_{l-1}, 1)$ (see~\cite{Y24} for the definition
of the $\kappa$th partition function) 
with $Z^{\{0\}}$ 
being the topological partition function of $M(A_{l-1}, 1)$.
Write $\log \mathcal{D}(\bT;\e;l)=:\sum_{g\ge0} \e^{2g-2} \F^{\{1\}}_g(\bT;l)$. Then 
by \cite[Theorem 6]{Y24} we know that for $g=0$,
$\F^{\{1\}}_0(\bT;l)$ has the following expression
\beq\label{FF10}
\F^{\{1\}}_0(\bT;l) = \frac12\sum_{a,b=0,\dots,l-1}\sum_{j_1,j_2\ge0} (T^{a,j_1}-\delta^{a,1}\delta^{j_1,0})(T^{b,j_2}-\delta^{b,1}\delta^{j_2,0}) \Omega^{\mathbb{P}^1_{l-1,1}}_{a,j_1;b,j_2}({\bf v}^{\{1\}}(\bT;l)),
\eeq
where $\Omega^{\mathbb{P}^1_{l-1,1}}_{a,j_1;b,j_2} $ 
are genus~0 two-point correlation functions for $M(A_{l-1}, 1)$ and 
$
{\bf v}^{\{1\}}(\bT;l)=(v^{0,\{1\}}(\bT;l),\dots,v^{l-1,\{1\}}(\bT;l))
$
 is the $\kappa$th solution to the principal hierarchy of $M(A_{l-1}, 1)$ with $\kappa=1$, 
and that for $g\ge1$,
\beq\label{FF1g}
\F^{\{1\}}_g(\bT;l)=F_g^{\mathbb{P}^1_{l-1,1}} 
\biggl({\bf v}^{\{1\}}(\bT;l), \frac{\p {\bf v}^{\{1\}}(\bT;l)}{\p T^{0,0}}, \dots, \frac{\p^{3g-2} {\bf v}^{\{1\}}(\bT;l)}{\p (T^{0,0})^{3g-2}}\biggr),
\eeq
where $F_g^{\mathbb{P}^1_{l-1,1}} $ denotes the genus $g$ free energy of $M(A_{l-1},1)$ in the jet-variables~\cite{DZ-norm}.

We recall that $Z^{\{\kappa\}}$ for any $\kappa$ 
is a tau-function for the Dubrovin--Zhang hierarchy of $M(A_{l-1}, 1)$.
It is proved in~\cite{CLPS} (cf.~\cite{BW}) that the Dubrovin--Zhang hierarchy of $M(A_{l-1}, 1)$ is Miura equivalent to the 
extended bigraded Toda hierarchy, which is defined~\cite{Carlet} by
\begin{align}
&\e\frac{\p L}{\p T^{0,j}}=\frac{2}{j!}\Bigl[\bigl(L^{j}(\log L-\frac{l}{2l-2}c_j)\bigr)_+,L\Bigr],\quad j\ge0,\\
&\e\frac{\p L}{\p T^{a,j}}=\frac{1}{(l-1)(\frac{a}{l-1})_{j+1}}\Bigl[\bigl(L^{\frac{a}{l-1}+j}\bigr)_+,L\Bigr],
\quad j\ge0, a=1,\dots,l-1,\\
&\e\frac{\p L}{\p T^{l-1,j}}=\frac{1}{(j+1)!}\Bigl[\bigl(L^{j+1}\bigr)_+,L\Bigr],
\quad j\ge0,
\end{align}
where $L$ is the Lax operator given by 
\beq\label{definitionL}
L:= \TT^{l-1} + u_{l-2} \TT^{l-2} + \dots +u_1 \TT +u_{0} + u_{-1} \TT^{-1}, \quad {\rm with}~ \TT:=e^{\e \p_x},
\eeq
 $c_j=\sum_{i=1}^j\frac1i$, $j\ge0$, and for the details about $\log L$ see~\cite{Carlet, CDZ}.
 Note that here the normalization of flows is the one suggested by Gromov--Witten theory.
Using the string equation~\eqref{string} we will prove in Section~\ref{review} the following lemma.
\begin{lemma}\label{lemma-initial}
The initial value for the solution ${\bf u}^{\{1\}}={\bf u}^{\{1\}}(\bT;\e;l)$ of the extended bigraded Toda hierarchy 
corresponding to $\mathcal{D} = Z^{\{1\}}$ is given by
 \beq\label{initial}
u^{\{1\}}_{-1}({\bf T};\e;l) |_{T^{a,j}=x\delta^{a,0}\delta^{j,0}}  \equiv x, \quad u^{\{1\}}_{i}({\bf T};\e;l) |_{T^{a,j}=x\delta^{a,0}\delta^{j,0}} \equiv 0  ~(i=0,\dots, l-2).
\eeq
\end{lemma}
 
The above-mentioned Carlet--van de Leur--Posthuma--Shadrin theorem explicitly says that 
with a careful choice of the calibration of Frobenius manifold,
\beq\label{CLPSthm}
Z^{\rm hyperone}|_{x=1}=\mathcal{D}(\bT;\e;l)|_{\bT=\mathbf{T}(1,{\bf s})},
\eeq
where 
\beq
\mathbf{T}(x,{\bf s}):=(T^{l-1,j}=(j+1)! s_{j+1}, 
T^{a,j}=x\delta^{a,0}\delta^{j,0}, a=0,\dots,l-2, \, j\ge0).
\eeq
and
the equality is understood in the sense that it holds up to multiplication by a constant factor that can depend on~$\e$.
By using a scaling symmetry, Lemma~\ref{lemma-initial} and an argument in~\cite{DY1} (cf.~\cite{Du2, Yang, YZ}) we can 
deduce the refined version of the Carlet--van de Leur--Posthuma--Shadrin theorem as follows: under a careful choice of calibration 
of Frobenius manifold, 
\beq\label{refinedCvLPS} 
Z(x,{\bf s};\e;l)=\mathcal{D}(\bT;\e;l)|_{\bT=\mathbf{T}(x,{\bf s})},
\eeq
where
we recall that $Z$ is the corrected partition function.
Identity~\eqref{refinedCvLPS} holds up to multiplying by a pure constant only that is usually irrelevant of the study (this 
is simply because on the right-hand side we do not have a canonical choice of this 
pure constant coming from the genus~1 free energy).
We note that the study of $k$-point functions for $\mathcal{D}(\bT;\e;l)$ 
with all but the extended flows will be postponed to 
another publication.

Using~\eqref{refinedCvLPS} and \eqref{FF10}, \eqref{FF1g}, we obtain the following theorem.  
\begin{theorem}\label{theorem2orbi}
The following formulas hold:
\begin{align}
&\mathcal{F}_0(x,\mathbf{s};l)
=\frac{1}{2}\Bigl(\,\sum_{j_1,j_2\ge0}(j_1+1)!(j_2+1)!s_{j_1+1}s_{j_2+1}\Omega^{\mathbb{P}^1_{l-1,1}}_{l-1,j_1;l-1,j_2}({\bf v}^{\{1\}}(\bT;l))\nn\\
 &+2\sum_{j\ge0}(j+1)!s_{j+1}\bigl(x\Omega^{\mathbb{P}^1_{l-1,1}}_{0,0;l-1,j}({\bf v}^{\{1\}}(\bT;l))
 -\Omega^{\mathbb{P}^1_{l-1,1}}_{1,1;l-1,j}({\bf v}^{\{1\}}(\bT;l))\bigr)\nn\\
 &+x^2
 v^{l-1,\{1\}}(\bT;l)-2x
  \Omega^{\mathbb{P}^1_{l-1,1}}_{0,0;1,1}({\bf v}^{\{1\}}(\bT;l))
 +\Omega^{\mathbb{P}^1_{l-1,1}}_{1,1;1,1}({\bf v}^{\{1\}}(\bT;l))\Bigr)\Big|_{\bT=\mathbf{T}(x,{\bf s})},
\end{align}
\begin{align}
&\mathcal{F}_g(x,\mathbf{s};l)=F_g^{\mathbb{P}^1_{l-1,1}} 
\Bigl({\bf v}^{\{1\}}(\bT;l), \frac{\p {\bf v}^{\{1\}}(\bT;l)}{\p T^{0,0}}, \dots, \frac{\p^{3g-2} {\bf v}^{\{1\}}(\bT;l)}{\p (T^{0,0})^{3g-2}}\Bigr)\Big|_{\bT=\mathbf{T}(x,{\bf s})}, \quad g\ge1.
\end{align}
\end{theorem}

For the case when $l=2$, Theorem~\ref{theorem2orbi} was given in~\cite[Example 23]{Du2} (cf.~\cite[Theorem 1.1]{Yang}).

The paper is organized as follows. 
In Section~\ref{review}, we do some preparations.
In Section~\ref{conclusion}, we prove Theorems~\ref{cnj1}, \ref{curve}.
In Section~\ref{compute}, 
we do concrete computations based on Theorem~\ref{cnj1}.

\medskip

\noindent {\bf Acknowledgements}
The work was partially supported by NSFC No. 12371254 and the CAS No. YSBR-032.

\section{Initial value and wave functions}\label{review}
In this section, we prove Lemma~\ref{lemma-initial} and then do some further preparations for next sections. 
 
Following~\cite{DYZ2, Y20}, we call 
$\psi=\psi(\lambda,x,\e;l)$
a {\it wave function associated to $L_{\rm ini}$} if 
 \beq\label{wave-eq}
L_{\rm ini}\psi=\lambda\psi,
\eeq
where $L_{\rm ini}$ is the initial Lax operator, that is the 
Lax operator $L$~\eqref{definitionL} being restricted at the initial value. 
An element $\psi_A(\lambda,x,\e;l)$ of the form
\beq
\psi_A(\lambda,x,\e;l)=\lambda^{\frac{x}{(l-1)\e}}\bigl(1+O(\lambda^{-\frac{1}{l-1}})\bigr)
\in\lambda^{\frac{x}{(l-1)\e}}\otimes\CC((\e))[[x]]((\lambda^{-\frac{1}{l-1}}))
\eeq
is called a {\it (formal) wave function of type A associated to $L_{\rm ini}$} if it satisfies~\eqref{wave-eq}.
An element $\psi_B(\lambda,x,\e)$ of the form
\beq
\psi_B(\lambda,x,\e;l)=e^{-s(x,\e;l)}\lambda^{-\frac{x}{\e}}(1+O(\lambda^{-1}))
\in e^{-s(x,\e;l)}\lambda^{-\frac{x}{\e}}\otimes\CC((\e))[[x]]((\lambda^{-1})),
\eeq
with $s(x,\e;l)=-\bigl(1-\TT^{-1})^{-1}(\log u_{-1}(\mathbf{T};\e;l) |_{T^{a,j}=x\delta^{a,0}\delta^{j,0}}\bigr)$, is called a {\it (formal) wave function of type B associated to $L_{\rm ini}$} if it satisfies~\eqref{wave-eq}. 

Now we restricted to the consideration of the special solution of the bigraded Toda hierarchy corresponding to $l$-hypermaps for any fixed $l\ge2$.
\begin{proof}[Proof of Lemma~\ref{lemma-initial}]
Similarly to~\cite{BDY-DS}, one can 
deduce from the string equation~\eqref{string} that 
\beq
\frac{\p^2\log Z^{\{1\}}}{\p T^{1,0}\p T^{a,0}}\bigg|_{T^{b,j}=x\delta^{b,0}\delta^{j,0}}=\frac{x}{\e^2}\delta_{a,l-1},\quad a=0,\dots,l-1.
\eeq
Then by an argument similar to~\cite{BDY-DS} and using the definition of tau-function given in~\cite{Carlet} the lemma is proved.
\end{proof}

In Section~\ref{section1} we recalled a family of algebraic curves from~\cite{DYZ3}. We will now review another 
family of algebraic curves. Originally, the two families of algebraic curves were introduced~\cite{DYZ3} to study certain $r$-spin invariants.
Following~\cite{DYZ3}, define 
\beq
w=1+u-\frac{r-1}6 u^2+\frac{(r-1)(2r+1)}{72}u^3-\cdots
\eeq
as the unique power-series solution to
\beq\label{29}
\frac{w^{r+1}}{r(r+1)}-\frac{w}{r}+\frac{1}{r+1}=\frac{u^2}{2},
\eeq
and define $C_{\ell}(r,j)\in\mathbb{Q}[r,j], \ell\ge0$, by means of generating series as follows:
\beq\label{30}
\sum_{\ell\ge0}C_{\ell}(r,j)u^{\ell+1}=
\left\{
\begin{aligned}
\frac{w^{j+1}-1}{j+1},\quad j\neq-1,\\
\log w,\quad j=-1.
\end{aligned}\right.
\eeq
The first few values are given by
\begin{align}
&C_{0}(r,j)=1,\quad C_1(r,j)=\frac{j}2-\frac{r-1}6,\quad C_{2}(r,j)=\frac{j(j-r)}{6}+\frac{(r-1)(2r+1)}{72},\\
&C_3(r, j) =\frac{ j(j - r)(j - r - 1)}{24}-\frac{(r - 1)(r + 2)(2r + 1)}{540}.
\end{align}

The following two identities are proved in~\cite[Lemma 1]{DYZ3}:
\begin{align}
&f_{r,j+1}(T)=\Bigl(1+\bigl(\frac{r-1}{2}-j\bigr)T+(r+1)T^2\partial_T\Bigr)f_{r,j}(T),\label{psiA-1}\\
&f_{r,j+r}(T)=f_{r,j}(T)-rjTf_{r,j-1}(T),\label{psiA-2}
\end{align}
where $f_{r,j}(T)\in \mathbb{Q}[r,j][[T]]$ are power series of~$T$ defined by
\beq
f_{r,j}(T):=\sum_{\ell\ge0}(2\ell+1)!!C_{2\ell}(r,j)(-T)^\ell.
\eeq

The following lemma generalizes a result in~\cite{Y20}.
\begin{lemma}\label{lemma-psi}
The two elements $\psi_A(\lambda,x,\e;l)$ and $\psi_B(\lambda,x,\e;l)$ defined by 
\begin{align}
\psi_A(\lambda,x,\e;l)&=\lambda^{\frac{x}{(l-1)\e}}f_{l-1,\frac x\e}\Bigl(\frac{\e}{(l-1)\lambda^{\frac{l}{l-1}}}\Bigr),\label{psiA}\\
\psi_B(\lambda,x,\e;l)&=\Gamma(\frac x\e+1)\e^{\frac x\e}\lambda^{-\frac x\e}\sum_{m\ge0}\frac{(\frac x\e+1)_{lm}\,\e^{(l-1)m}}{l^m m!\lambda^{lm}}\label{psiB}
\end{align}
are formal wave functions of type A  
 and of type B associated to $L_{\rm ini}=\TT^{l-1}+x\TT^{-1}$, respectively.
\end{lemma}
\begin{proof}
For $\psi_A(\lambda,x,\e;l)$, the statement follows from~\eqref{psiA-2}. 
Noting that in our situation, 
\beq
s(x,\e;l)=-\log\Gamma(\frac x\e+1)-\frac x\e\log\e,
\eeq
the statement for $\psi_B(\lambda,x,\e;l)$ follows by a direct verification.
\end{proof}

We note that $\psi_B$ is related to the quantum spectral curve for $l$-hypermaps
\cite{DM, GS, MulaseS} (cf.~\cite{DYZ3}), and a result essentially equivalent to formula~\eqref{psiB} was obtained in~\cite[Proposition~10]{DM}.

The following relations for $\widetilde C_s(r,i,j)$ (see~\eqref{deftc1}), which are proved in~\cite{DYZ3}, 
\begin{align}
& \widetilde C_s(r,i+1,j) - \widetilde C_s(r,i,j+1) = 2  \widetilde C_{s-1}(r,i,j) ,  \label{tcid1} \\
& \widetilde C_s(r,i+r+1,j)-\widetilde C_s(r,i,j+r+1) = 2(r+1) \widetilde C_{s-1}(r,i,j),  \label{tcid2} \\
& f_{r,i}(T)f_{r,j}(-T) = \sum_{s\geq0}  \bigl(1+\tfrac{s-i-j-1}r\bigr)_s \widetilde C_s(r,i,j)\Bigl(\frac{rT}2\Bigr)^s\label{fifjTmT}
\end{align}
will also be useful.

\section{Proof of Theorems~\ref{cnj1} and~\ref{curve}}\label{conclusion}

In this section we prove Theorems~\ref{cnj1} and~\ref{curve}.

Define~\cite{FYZ} the {\it matrix Lax operator} $\mathcal{L}$ by
\beq
\mathcal{L}:=\TT+\Lambda(\lambda)+V,
\eeq
where 
$\Lambda(\lambda)=-\lambda e_{1,l-1}-\sum_{i=2}^{l}e_{i,i-1}$
 and 
$ V=\sum_{j=1}^{l}u_{l-1-j}e_{1,j}$.
Here,  
 $e_{i,j}$ denotes the $l\times l$ matrix with the $(i,j)$-entry being~$1$ and others $0$.
Let $\mathcal{L}_{\rm ini}:=\mathcal{L}|_{T^{a,j}=\delta^{a,0}\delta^{j,0}x}$ with the initial data given by~\eqref{initial}. 
We have $\mathcal{L}_{\rm ini}=\TT+\Lambda(\lambda)+V(x)$, where 
$V(x)=xe_{1,l}$.

Define a matrix $\Psi(\lambda,x,\e;l)=(\Psi(\lambda,x,\e;l)_{i,j})_{i,j=1,\dots,l}$ by
\beq\label{Psi}
\Psi(\lambda,x,\e;l)_{i,j}:=\left\{
\begin{array}{ll}
\psi_A(e^{2\pi\sqrt{-1}(j-1)}\lambda,x+(l-1-i)\e,\e;l),& j=1,\dots,l-1,\\
\psi_B(\lambda,x+(l-1-i)\e,\e;l),& j=l.
\end{array}
\right.
\eeq
By using Lemma~\ref{lemma-psi} we arrive at
\begin{prop}\label{bispPsi}
The following equations hold:
\begin{align}
&\mathcal{L}_{\rm ini}\Psi(\lambda,x,\e;l)=0.\label{Translation-Psi}
\end{align}
\end{prop}

We also introduce two functions $\psi^*_A(\lambda,x,\e;l), \psi^*_B(\lambda,x,\e;l)$ by
\begin{align}
\psi^*_A(\lambda,x,\e;l)&=\frac1{l-1}\lambda^{-\frac{x}{(l-1)\e}-1}
f_{l-1,-\frac x\e-1}\Bigl(\frac{-\e}{(l-1)\lambda^{\frac{l}{l-1}}}\Bigr),\\
\psi^*_B(\lambda,x,\e;l)&=-\frac{\e^{-\frac{x}{\e}}}{\Gamma(\frac x\e+1)}\lambda^{\frac x\e-1}\sum_{m\ge0}\frac{(-1)^m(\frac x\e+1-lm)_{lm}\,\e^{(l-1)m}}{l^m m!\lambda^{lm}},
\end{align}
and a matrix $\Psi^*(\lambda,x,\e;l)=(\Psi^*(\lambda,x,\e;l)_{i,j})_{i,j=1,\dots,l}$ by
\beq
\Psi^*(\lambda,x,\e;l)_{i,j}=\left\{
\begin{array}{lll}
\psi^*_A(e^{2\pi \sqrt{-1}(i-1)}\lambda, x-j\e,\e;l),&i=1,\dots,l-1,&j=1,\dots,l-1,\\
- x\psi^*_A(e^{2\pi \sqrt{-1}(i-1)}\lambda, x,\e;l),&i=1,\dots,l-1,&j=l,\\
\psi^*_B(\lambda, x-j\e,\e;l),&i=l,&j=1,\dots,l-1,\\
-x\psi^*_B(\lambda, x,\e;l),&i=l,&j=l.\\
\end{array}
\right.
\eeq
We note that the functions $\psi^*_{A}, \psi^*_B$ both satisfy
\beq
(\TT^{-(l-1)}+(x+\e)\TT)\psi^*=\lambda\psi^*.
\eeq
Here we point out that $\TT^{-(l-1)}+(x+\e)\TT$ is the adjoint operator of $L_{ini}$ (cf.~\cite{Dickey,DYZ2,DYZ3}).
\begin{prop}
The following identity holds:
\beq
\Psi(\lambda,x,\e;l)\Psi^*(\lambda,x,\e;l)\equiv I_l.
\eeq
\end{prop}

\begin{proof}
By using the definitions of $\Psi(\lambda,x,\e;l)$, $\Psi^*(\lambda,x,\e;l)$ and by using~\eqref{fifjTmT}, we find that 
the proposition is equivalent to the validity of the following identity: 
\begin{align}
&\widetilde C_{j-i+(l-1)(m+1)}(l-1,n+l-1-i,-n+j-1)\nn\\
& \equiv \frac{2^{j-i+(l-1)(m+1)}}{l^{m}}
\sum_{t=0}^{m}(-1)^{m-t}\binom{m}{t}\binom{n+l(t+1)-i-1}{l(m+1)+j-i-1}\label{tCpre}
\end{align}
for all $n\in\CC$ and for all $i=1,\dots,l$, 
$j=0,\dots,l-1$, 
$m\in\mathbb{Z}$ satisfying $j-i+(l-1)(m+1)\ge0$.
To prove~\eqref{tCpre} let us prove the following stronger one:
\beq\label{tCfin}
\widetilde C_{s}(l-1,y,-y+s-(l-1)m-1)
\equiv\frac{2^{s}}{l^{m}}
\sum_{t=0}^{m}(-1)^{m-t}\binom{m}{t}\binom{y+lt}{s+m}
\eeq
for all $y\in\CC$ and for all $s, m\ge0$. 
By a straightforward calculation one can show the validity of~\eqref{tCfin} 
for the case when $s\ge0, m=0$ and $y=0$. 
Then by using~\eqref{tcid1} and the known polynomiality in~$y$ on both sides we can prove~\eqref{tCfin} 
for $s\ge0, m=0$ and $y\in\CC$. 
Further using~\eqref{tcid2} we get the validity for $s,m\ge0$ and $y\in\CC$.
The proposition is proved.
\end{proof}

Let $R_a(\lambda)$, $a=1,\dots, l-1$, be 
the basic matrix resolvents of $\mathcal{L}$~\cite{FYZ, HY}. 
Denote by $\mathcal{R}_a(\lambda,x,\e;l)$ these $R_a(\lambda)$ evaluated at the initial value~\eqref{initial}. 

Introduce a diagonal matrix
\beq
P_a(l):=\diag(1,\xi_{l-1}^a,\dots,\xi_{l-1}^{(l-2)a},0),\quad a=1,\dots,l-1,
\eeq
where $\xi_{l-1}:=e^{\frac{2\pi\sqrt{-1}}{l-1}}$.
Similar to~\cite{BDY1, DYZ1, HY} let us prove the following lemma. 
\begin{lemma} \label{bispRM}
For $a=1,\dots,l-1$, we have
\beq\label{MR-TE}
\mathcal{R}_a(\lambda,x,\e;l) =  
\lambda^{\frac{a}{l-1}} 
\Psi(\lambda,x,\e;l)
 P_a(l)
\Psi(\lambda,x,\e;l)^{-1}.
\eeq
\end{lemma}
\begin{proof}
The lemma is proved by using~\eqref{Translation-Psi}, \eqref{Psi} and \cite[Lemma~7]{HY}.
\end{proof}

\begin{proof}[Proof of Theorems \ref{cnj1} and~\ref{curve}]
Taking $\bT=\mathbf{0}$ in~\cite[Proposition 1.6]{FYZ} and using Lemma~\ref{bispRM} with 
 the notice that $\mathcal{R}_{m_1}(\lambda,n,1;l)=\lambda M(\lambda,n;l)$, 
and using the Carlet--van de Leur--Posthuma--Shadrin theorem~\eqref{CLPSthm}, we get~\eqref{npoint},
 \eqref{Mm1}, \eqref{gm1} and~\eqref{ynew}
(it is proved in~\cite{HY} that the tau-structure from~\cite{FYZ} is the same as the one in~\cite{Carlet}).
Similar to the method in~\cite{DY1}, we 
get~\eqref{1point}, \eqref{1point2}.
\end{proof}

From~\eqref{tCfin} we can prove directly the 
equivalence between~\eqref{1point} and~\eqref{1point2}.

Let us now give direct proofs of equivalence between~\eqref{1point2} and Zagier's formula~\eqref{Zagierformula2} as well as of equivalence between~\eqref{1point2} and Zagier's formula~\eqref{Zagierformula}.
Formula~\eqref{1point2} tells that for $n\in\CC$, 
\begin{align}
\sum_{g\ge0}lmM^{[l]}_{g,1}(lm)n^{1-2g+(l-1)m}
&=-\frac{(ml)!}{2^{(l-1)m} m!}{\rm res}_{u=0}(X+1)^n(X-1)^{-n}\frac1{u^{(l-1)m}}\frac{dX}{du}du\nn\\
&=-\frac{(ml)!}{l^m m!}{\rm res}_{Y=1}Y^{-n}\frac{(1-Y^l)^m}{(1-Y)^{ml+2}}dY,\label{midZagierformula}
\end{align}
where in the second equality we performed the change of variable 
$X=(1+Y)/(1-Y)$.
Taking coefficients of powers of~$n$ in~\eqref{midZagierformula}
we obtain
\begin{align}
lmM^{[l]}_{g,1}(lm)=-\frac{(lm)!}{l^m m!}{\rm res}_{Y=1}\frac{(-\log Y)^{1-2g+(l-1)m}}{(1-2g+(l-1)m)!}\frac{(1-Y^l)^m}{(1-Y)^{lm+2}}
dY\nn\\
=\frac{(lm)!}{l^m m!(1-2g+(l-1)m)!}{\rm res}_{t=0}t^{1-2g+(l-1)m}\frac{(1-e^{-lt})^m}{(1-{e^{-t}})^{lm+2}}e^{-t}
dt,
\end{align}
which gives Zagier's formula~\eqref{Zagierformula2}.
For $n\in\mathbb{Z}_{\ge0}$, using the residue theorem based on~\eqref{midZagierformula}, we get
\begin{align}
\sum_{g\ge0}lmM^{[l]}_{g,1}(lm)n^{1-2g+(l-1)m}
&=\frac{(ml)!}{l^m m!}{\rm res}_{Y=0}Y^{-n}\frac{(1-Y^l)^m}{(1-Y)^{ml+2}}dY,
\end{align}
which gives Zagier's formula~\eqref{Zagierformula}.

\section{Computations}\label{compute}
In this section, we do some computations based on~\eqref{1point}, \eqref{npoint}.

It will be convenient to introduce a function 
$f(i,j,p;l)$ of $i,j\in \CC$, $p\in \ZZ$, by
\beq\label{def_f}
f(i,j,p;l) :=
\frac1{l^{p}p!}
\sum_{s=0}^{p}
(-1)^{s}
\binom{p}{s}
(i+1-ls)_{j-1}.
\eeq

Consider $k=1$. Recall by definition that the number $M^{[l]}_{g,1}(a+1)$ vanishes if $l {\not|}\, a+1$. From~\eqref{1point} we obtain that when 
$a+1$ is a multiple of $l$, 
\beq
(a+1)M^{[l]}_1(a+1;n)=
\frac1{a+2} \, f\Bigl(n-1,a+3,\frac{a+1}{l};l\Bigr). \label{onepointexplicit}
\eeq
Here, $a\ge0$.
For the special case when $l=2$, the function $f(i,j,p;l)$ for $j$ being a positive integer can be alternatively written as 
\beq\label{equivlence}
f(i,j,p;2)=\frac{(j-1)!}{2^pp!}\sum_{s=0}^p2^{s}\binom{p}{s}
\binom{i+j-2p-1}{s+j-2p-1},
\eeq
from which one easily sees the equivalence between~\eqref{onepointexplicit} and \cite[(3.2.4)]{DY1}.

Taking $g=0$ in Zagier's formula~\eqref{Zagierformula2}, we obtain
\beq\label{genus0M}
lmM^{[l]}_{0,1}(lm)=\frac1{lm+1} \binom{lm+1}{m}.
\eeq
Formula~\eqref{genus0M} was obtained in~\cite[Proposition 13 and Remark 14]{DM}.
Taking $g=1,2,3$ in Zagier's formula~\eqref{Zagierformula2}, we obtain
\begin{align}
&lmM^{[l]}_{1,1}(lm)=lm\binom{lm-1}{m}\frac{(l-1)lm-2}{24},\\
&lmM^{[l]}_{2,1}(lm)=(lm-2)_3\binom{lm-3}{m}\frac{5(l-1)^2l^2m^2-(2l^4+20l^2-22l)m+24}{5760},\\
&lmM^{[l]}_{3,1}(lm)=(lm-4)_5\binom{lm-5}{m}\nn\\
&\times \frac{35 (l-1)^3 l^3 m^3-42 (l-1)^2 l^2 (l^2+l+6) m^2+(16 l^6+84 l^4+504 l^2-604 l) m-480}{2903040}.
\end{align}

When $l$ is odd, comparing coefficients of~$n^1$ on both sides of~\eqref{1point}, we obtain
\beq
lmM^{[l]}_{\frac{(l-1)m}{2},1}(lm)=\frac{(lm)!}{l^m m! (lm+1)}
\sum_{s=0}^m\binom{m}{s}\Big/\binom{lm}{ls}.
\eeq
When $l$ is even, assuming $m$ is even, we obtain from~\eqref{1point} that
\beq
lmM^{[l]}_{\frac{(l-1)m}{2},1}(lm)=\frac{(lm)!}{l^m m! (lm+1)}
\sum_{s=0}^m(-1)^s\binom{m}{s}\Big/\binom{lm}{ls}.
\eeq

For $k\ge2$, using~\eqref{npoint} one can easily compute out 
\begin{align*}
&M^{[3]}_{2}(7,8;n)=7 n^2 (9 n^8+600 n^6+11077 n^4+55050 n^2+47664),\\
&M^{[4]}_{2}(2,6;n)=\frac{5}{2} n^2 (n^4+16 n^2+25),\\
&M^{[5]}_3(3,3,4;n)=8 n (2 n^6+43 n^4+161 n^2+46),\\
&M^{[6]}_4(1,3,3,5;n)=40 n^2 (5 n^6+205 n^4+1612 n^2+1866).
\end{align*}
Similar to \cite{BDY1, DY1, DYZ2}, using~\eqref{npoint} it is straightforward to obtain formulas like 
\begin{align}
&\sum_{b\ge1} bM^{[l]}_{2}(2,b;n) \frac1{\lambda^{b+1}}= 2n\lambda g(\lambda,n,l,0) - n g(\lambda,n,l,l-1) - n g(\lambda,n,1,0) -2\lambda, 
\label{special2b}
\end{align}
where $g$ is given by~\eqref{gm1}.
An algorithm that leads to this kind of formulas systematically is given in \cite{DY1} (see also~\cite{HY}), 
that is also very efficient for computing $M^{[l]}_{g,k}(b_1,\dots,b_k)$ with equaling $b$'s (see Section~\ref{section4.2}).
We note that formula~\eqref{special2b} agrees with~\cite[Theorem 3.11]{BC}. 
By a Laurent expansion we obtain from~\eqref{npoint} that when 
$a+b+2$ is a multiple of~$l$, 
\begin{align}
&(a+1)(b+1)M^{[l]}_2(a+1,b+1;n)
=
n^2\sum_{0\le j\le a\atop l| j-a}(j+1)
f(n,a-j,\frac{a-j}{l};l)
f(n,j+b+2,\frac{j+b+2}{l};l)
\nn\\
&-n\sum_{i_1=1}^{l-1}\sum_{0\le j\le a\atop l| a-j+i_1}(j+1)
f(n,a-j,\frac{a-j+i_1}{l}-1;l)
f(n-i_1,j+b+2,\frac{j+b+2-i_1}{l};l)
\nn\\
&-n\sum_{i_2=1}^{l-1}\sum_{0\le j\le a\atop l| a-j-i_2}(j+1)
f(n-i_2,a-j,\frac{a-j-i_2}{l};l)
f(n,j+b+2,\frac{j+b+2+i_2}{l}-1;l)
\nn\\
&+\sum_{i_1,i_2=1}^{l-1}\sum_{0\le j\le a\atop l | i_1-i_2+a-j}(j+1)
f(n-i_2,a-j,\frac{a-j+i_1-i_2}{l}-1;l)
f(n-i_1,j+b+2,\frac{j+b+2-i_1+i_2}{l}-1;l), \label{twopontnumber}
\end{align}
and $M^{[l]}_{2}(a+1,b+1;n)$ vanishes if $l {\not|}\, a+b+2$. 
Here, $a,b\ge0$. 
Again with the help of~\eqref{equivlence} it is easy to see that for the case when $l=2$ formula \eqref{twopontnumber} becomes \cite[formula (3.2.5)]{DY1}.

\subsection{Comparison with Grothendieck's dessin counting}

Denote by $\mathcal{P}_d$ the set of partitions of weight $d$. 
Let $g, d$ be nonnegative integers and let $\mu, \nu \in \mathcal{P}_d$. 
Recall from~\cite{ALS16} (see also~\cite{GGM1,GGM2,GGM3}) that the {\it strictly monotone double Hurwitz
number $h_g(\mu, \nu)$ in genus $g$ and degree $d$} is defined as the number of tuples $(\alpha, \tau_1,\dots, \tau_r, \beta)$ satisfying
\begin{itemize}
\item[(i)] $\alpha, \beta$ are permutations of $\{1,\dots,d\}$ with cycle type $\mu, \nu$ respectively, and $\tau_1,\dots , \tau_r$
are transpositions such that $\alpha\tau_1\dots\tau_r = \beta$,
\item[(ii)] the subgroup generated by $\alpha, \tau_1,\dots, \tau_r$ acts transitively on $\{1,\dots,d\}$,
\item[(iii)] writing $\tau_j = (a_j, b_j )$ with $a_j < b_j$, $j = 1,\dots, r$, then we require $b_1 <\dots< b_r$.
\end{itemize}
Here $r = \ell(\mu) + \ell(\nu) + 2g-2$.
According to~\cite{ALS16}, 
when all components of~$\mu$ equal~$l$, i.e., $\mu=(l,\dots,l)$, 
counting the corresponding strictly monotone Hurwitz number is equivalent to counting $l$-hypermaps with $(b_1,\dots,b_k)=\nu$. 
Explicitly, by~\cite[Proposition~4.8]{ALS16} we have
\beq\label{Mgk-hur}
M^{[l]}_{g,k}(\nu_1,\dots,\nu_k)=\frac{\prod_{i=1}^{\infty}n_i(\nu)!}{|\nu|!}h_g(\mu,\nu).
\eeq

Let us now recall Grothendieck's dessin counting in the sense of e.g.~\cite{KZ,LZ,YZ,Zhou2,Zograf}. 
Let $(C,f)$ be a Belyi pair of genus $g$ and degree $d$, and $\Gamma$ the corresponding dessin. Put $v=|f^{-1}(0)|$, $k=|f^{-1}(1)|$ and $m=|f^{-1}(\infty)|$. Assume the poles of $f$ are labelled and denote the set of their orders by $\mu=(\mu_1,\dots,\mu_m)$ so that $\sum_{i=1}^m\mu_i=d$. The triple $(v,k,\mu)$ is called the type of the dessin $\Gamma$, and the set of all dessins of type $(v,k,\mu)$ will be denoted by $\mathcal{D}_{v,k,\mu}$. We call 
\beq
N_{v,k}(\mu_1,\dots,\mu_m):=\sum_{\Gamma\in\mathcal D_{v,k,\mu}}\frac1{|{\rm Aut}_b\Gamma|}
\eeq
{\it Grothendieck's dessin counting}, i.e., 
the weighted count of labelled dessins. Here, 
 ${\rm Aut}_b\Gamma$ denotes the group of automorphisms of $\Gamma$ that preserve the boundary component wise.

In~\cite{CDO} (see also~\cite{GGR}), the strictly monotone Hurwitz numbers are shown to be related to 
connected LUE correlators. In~\cite{YZ} (cf.~\cite{AC}), the weighted count of dessins d'enfants is related to the connected LUE correlators. 
The following proposition is then obtained in~\cite{YZ}.

\noindent {\bf Proposition A (\cite{YZ}).} 
{\it For any $g\ge0$ and a partition $\mu= (\mu_1,\dots,\mu_m)$ of length $m$, and for $v,k,m\ge1$ satisfying
$|\mu|-m-l-k=2g-2$, denoting $n_i(\mu)$ the multiplicity of $i$ in $\mu$, we have
\beq\label{Nvk-hur}
N_{v,k}(\mu) =\frac{\prod_{i=1}^\infty n_i(\mu)!}{|\mu|!}\sum_{\nu\in\mathcal{P}_{|\mu|},\, \ell(\nu)=k}h_g(\mu,\nu).
\eeq}

For the case when $\mu=(l,\dots,l)$, by using~\eqref{Mgk-hur}
we see that \eqref{Nvk-hur} becomes
\beq\label{Nvk-l}
N_{v,k}(\mu)=\frac{\ell(\mu)!}{k!}\sum_{\nu\in\mathcal{P}_{d}, \, \ell(\nu)=k}
\binom{k}{n_1(\nu),\cdots,n_{|\mu|}(\nu)}M^{[l]}_{g,k}(\nu_1,\dots,\nu_k),
\eeq
where we call that $d=|\mu|$ and $g=1+(d-v-k-\ell(\mu))/2$.
Note that an explicit formula for generating series of $N_{v,k}(\mu)$ is given in~\cite[Proposition~1]{YZ}, from which 
we can compute the left-hand side of~\eqref{Nvk-l}, and 
using~\eqref{1point} and~\eqref{npoint} we can compute the right-hand side.
In this way we have checked the validity~\eqref{Nvk-l} for all $\mu$ with $|\mu|=d\le9$ and $\ell(\mu)\le4$, which in turn gives 
a verification of \eqref{1point} and~\eqref{npoint} for small cases.

Of course, formula~\eqref{Nvk-l} can be proved purely combinatorially,
which gives a partial answer to a question proposed in~\cite{YZ}. 
Indeed, from~\cite[Definition 3]{DM} we know that 
$M^{[l]}_{g,k}(\nu_1,\dots,\nu_k)$ is also the 
weighted count of connected 
genus $g$ and degree~$d$ branch covers of marked Riemann surfaces $f: (S; p_1, \dots, p_k) \to \mathbb{P}^1$ such that
$f$ is unramified over $\mathbb{P}^1\setminus \{0, 1, \infty\}$, 
each point in $f^{-1} (1)$ has ramification order~$l$, and the divisor $f^{-1}(\infty)$ equals $\nu_1p_1+\cdots+\nu_kp_k$.

 However, so far we still do not know other proofs for the more general Proposition~A.

\subsection{More computations}\label{section4.2}
In this subsection, we compute numbers like $M^{[l]}_{g,k}(b,\dots,b)$ by using 
an algorithm from~\cite{DY1} (cf.~\cite{HY}).
When $l=2$, concrete computations had been given in~\cite{DY1}. 
Let us consider $l\ge3$. 

For any $l\ge3$, when $b=1$ the number $M^{[l]}_{g,k}(1,\dots,1)$ vanishes unless $(g,k)=(0,l)$. For $(g,k)=(0,l)$, it is equal to counting circular permutation, giving
\beq
M^{[l]}_{g,k}(1,\dots,1)=(l-1)!\delta_{k,l}\delta_{g,0}.
\eeq

For $l=3$ and $b=3$, we list in Table~\ref{21tau22} the first few $M^{[3]}_{g,k}(3,\dots,3)$.
For $l=4$ and $b=4$, we list in Table~\ref{table44} the first few $M^{[4]}_{g,k}(4,\dots,4)$.
For $l=5$ and $b=5$, we list in Table~\ref{table55} the first few $M^{[5]}_{g,k}(5,\dots,5)$.
For $l=6$ and $b=5$, we list in Table~\ref{table65} the first few $M^{[6]}_{g,k}(5,\dots,5)$.

Finally, let us consider a duality between the numbers under consideration. 
Let $(\sigma_0,\sigma_1,\sigma_2)$ be an $l$-hypermap 
whose graphical representation
contains $k$ white $b$-gons and $k':=\frac{bk}{l}$ blue $l$-gons.
By exchanging white and blue we get a $b$-hypermap $(\sigma_0^{-1},\sigma_2^{-1},\sigma_1^{-1})$, 
whose graphical representation contains $k'$ white $l$-gons and $k$ blue $b$-gons.
For the purely combinatorial reason, we obtain
\beq\label{blue-whiteduality}
k'!M^{[l]}_{g,k}(b,\dots,b)=k!M^{[b]}_{g,k'}(l,\dots,l).
\eeq
We call this identity the {\it blue/white duality}. 
Geometric interpretation of this identity deserves a further study. 
Particularly, 
\beq\label{blue-white}
k'!M^{[l]}_{g,k}(2,\dots,2)=k!M^{[2]}_{g,k'}(l,\dots,l).
\eeq
Computations for the right-hand side of~\eqref{blue-white} (i.e., for $l$-gon angulations or say for $l$-valent ribbon graphs) 
were also given in~e.g.~\cite{DY1}.
We checked by means of explicit examples that our computations 
for the left-hand side of~\eqref{blue-white}
using the above-mentioned algorithm based on~\eqref{npoint} agree with those in~\cite{DY1}.

\begin{table}[ht]
\renewcommand{\arraystretch}{1.2}
$$
\begin{array}{|c|c|c|c|c|c|c|}
\hline
k&g=0&g=1&g=2&g=3&g=4\\
\Xhline{1pt}
1 &\frac13&\frac13&0&0&0\\
\hline
 2 & 1 & 3 & 0 & 0 & 0 \\
 \hline
 3 & 8 & \frac{152}{3} & 16 & 0 & 0 \\
 \hline
 4 & 112 & 1256 & 1416 & 0 & 0 \\
 \hline
 5 & 2304 & 41088 & 105984 & 22400 & 0 \\
 \hline
 6 & 63360 & 1670400 & 8012160 & 6044160 & 0 \\
 \hline
 7 & 2196480 & 81192960 & 640604160 & 1109038080 & 188160000 \\
 \hline
 8 & 92252160 & 4592931840 & 54935193600 & 177557452800 & 106424693760 \\
 \hline
\end{array}
$$
\caption{\label{21tau22}$M^{[3]}_{g,k}(3,\dots,3)$.}
\end{table}

\begin{table}[ht]
\renewcommand{\arraystretch}{1.2}
$$
\begin{array}{|c|c|c|c|c|c|c|}
\hline
k&g=0&g=1&g=2&g=3&g=4\\
\Xhline{1pt}
1&\frac14&\frac54&0&0&0\\
\hline
 2 & \frac{3}{2} & \frac{111}{4} & \frac{189}{4} & 0 & 0  \\
 \hline
 3 & 27 & 1170 & \frac{17307}{2} & \frac{18585}{2} & 0  \\
 \hline
 4 & 891 & \frac{145557}{2} & 1356183 & 6447033 & \frac{10584027}{2}  \\
 \hline
 5 & 44226 & 5989950 & 219806622 & 2674980450 & 9691264278  \\
 \hline
 6 & 2974320 & 613613880 & 38469942360 & 933634192860 & 8679911822820  \\
 \hline
 7 & 254304360 & 75230000640 & 7344595320300 & 307628585825220 & 5724003671541540 \\
 \hline
\end{array}
$$
\caption{\label{table44}$M^{[4]}_{g,k}(4,\dots,4)$.}
\end{table}

\renewcommand{\arraystretch}{1.2}
\begin{table}[ht]
$$
\begin{array}{|c|c|c|c|c|c|c|c|c|c|}
\hline
k&g=0&g=1&g=2&g=3&g=4\\
\Xhline{1pt}
1&\frac15&3&\frac85&0&0\\
\hline
 2 & 2 & 124 & 1210 & 1544 & 0  \\
 \hline
 3 & 64 & 9760 & 353664 & 3586592 & 8215040  \\
 \hline
 4 & 3840 & 1134720 & 97608576 & 3116265600 & 36432809856  \\
 \hline
 5 & 350208 & 174661632 & 28666036224 & 2030559252480 & 65131485161472  \\
 \hline
\end{array}
$$
\caption{\label{table55}$M^{[5]}_{g,k}(5,\dots,5)$.}
\end{table}

\begin{table}[ht]
\renewcommand{\arraystretch}{1.9}
$$
\begin{array}{|c|c|c|c|c|c|c|c|c|c|}
\hline
k&g=0&g=1&g=2&g=3\\
\Xhline{1pt}
 6 & 37950000&28538250000&9105774300000 &
 \begin{gathered}14983340\\61300000\end{gathered}\\
 \hline
 12& \begin{gathered}54227552841\\00000000000 \end{gathered}&
  \begin{gathered}2179514531294\\1000000000000 \end{gathered}&
 \begin{gathered} 43919127207074\\819700000000000 \end{gathered}&
  \begin{gathered}5644796628827485\\4554800000000000 \end{gathered}\\
  \hline
\end{array}
$$
\caption{\label{table65}$M^{[6]}_{g,k}(5,\dots,5)$ with $6|k$.}
\end{table}

\end{document}